\documentclass[letterpaper,12pt]{article}
\usepackage{amssymb, amsmath}
\usepackage{amsthm}
\usepackage{algorithm}
\usepackage{algorithmic}
\usepackage{enumerate}
\usepackage[numbers,square,comma,sort&compress]{natbib}
\usepackage[pdftex,colorlinks]{hyperref}
\usepackage{a4}

\newtheorem{theorem}{Theorem}

\newtheorem{lemma}[theorem]{Lemma}

\newcommand{\comment}[1]{}
 
\newcommand\bs[1]{\boldsymbol{#1}}

\begin{document}
\title{A deterministic sublinear-time
nonadaptive
algorithm for metric
$1$-median selection}

\author{
Ching-Lueh Chang \footnote{Department of Computer Science and
Engineering,
Yuan Ze University, Taoyuan, Taiwan. Email:
clchang@saturn.yzu.edu.tw}
\footnote{Innovation Center for Big Data and Digital Convergence,
Yuan Ze University, Taoyuan, Taiwan.}
\footnote{Supported in part by the Ministry of Science and Technology
of Taiwan under
grant
103-2221-E-155-026-MY2.}
}


\maketitle

\begin{abstract}
We
give
a deterministic $O(hn^{1+1/h})$-time $(2h)$-approximation
nonadaptive
algorithm for $1$-median selection in
$n$-point metric
spaces,
where
$h\in\mathbb{Z}^+\setminus\{1\}$ is arbitrary.
Our proof generalizes that of Chang~\cite{Cha13}.
\end{abstract}


\section{Introduction}

A metric space $(M,d)$ is a nonempty set $M$
endowed with a function $d\colon M\times M\to[\,0,\infty\,)$ such that
for all $x$, $y$, $z\in M$,
\begin{itemize}
\item $d(x,y)=0$ if and only if $x=y$,
\item $d(x,y)=d(y,x)$,
and
\item $d(x,y)+d(y,z)\ge d(x,z)$ (triangle inequality).
\end{itemize}
The
{\sc metric $1$-median}
problem asks for a point in
an $n$-point
metric space
$(M,d)$
with the minimum average distance to other points.
For $c\ge 1$,
a point
$\hat{p}\in M$
is said to be $c$-approximate for {\sc metric $1$-median} if
$$
\sum_{x\in M}\,d\left(\hat{p},x\right)
\le c\cdot
\min_{p\in M}\,
\sum_{x\in M}\,d\left(p,x\right).
$$
An algorithm for {\sc metric $1$-median} is nonadaptive
if the sequence of distances that it inspects
depends only on $M$ but not on $d$.
Because there are $n(n-1)/2$ nonzero distances,
``sublinear-time'' will mean ``$o(n^2)$-time.''

Indyk~\cite{Ind99, Ind00} shows that {\sc metric $1$-median} has a
Monte-Carlo $O(n/\epsilon^2)$-time $(1+\epsilon)$-approximation algorithm
for each $\epsilon>0$.
In $\mathbb{R}^D$,
where $D\ge1$,
{\sc metric $1$-median} has a
Monte-Carlo $O(2^{\text{poly}(1/\epsilon)}D)$-time
$(1+\epsilon)$-approximation algorithm for all $\epsilon>0$~\cite{KSS10}.
Many other algorithms are known for $k$-median
selection~\cite{GMMMO03, KSS10, ABS10}.
For example,
Guha et al.~\cite{GMMMO03}
give
a deterministic, $O(n^{1+\epsilon})$-time, $O(n^\epsilon)$-space,
$2^{O(1/\epsilon)}$-approximation
and one-pass
algorithm as well as a Monte-Carlo algorithm for $k$-median
selection in metric spaces, where $\epsilon>0$.


We show
that {\sc metric $1$-median} has a deterministic
$O(hn^{1+1/h})$-time $(2h)$-approximation nonadaptive algorithm for all
$h\in\mathbb{Z}^+\setminus\{1\}$, generalizing the following theorems:

\begin{theorem}[\cite{Cha13}]\label{mypreviousupperboundresult}
{\sc Metric $1$-median} has a deterministic $O(n^{1.5})$-time
$4$-approximation nonadaptive algorithm.
\end{theorem}

\begin{theorem}[\cite{Wu14}]
For each
$h\in\mathbb{Z}^+\setminus\{1\}$,
{\sc metric $1$-median} has a deterministic $O(hn^{1+1/h})$-time
$(2h)$-approximation (adaptive) algorithm.\footnote{The time complexity of
$O(hn^{1+1/h})$ is originally presented as $O(n^{1+1/h})$ because $h$ is
independent of $n$.
We include the $O(h)$ factor, which is implicit in the original proof,
for ease of comparison.}
\end{theorem}

When $n$ is a perfect square and $h=2$, our proof is equivalent to that of
Theorem~\ref{mypreviousupperboundresult}~\cite{Cha13}.
Chang~\cite{Cha14arXiv} shows that {\sc metric $1$-median} has no
deterministic $o(n^2)$-query $(4-\Omega(1))$-approximation algorithms
(where
an algorithm's
query complexity
is the number of distances that it inspects).
So the approximation ratio of $4$ in Theorem~\ref{mypreviousupperboundresult}
cannot be
improved
to
a
smaller
constant.

\section{Our algorithm}

Let $(\{0,1,\ldots,n-1\},d)$ be a metric space,
$h\in\mathbb{Z}^+\setminus\{1\}$ and $t=\lceil n^{1/h}\rceil$.
For all $j\in \{0,1,\ldots,n-1\}$,
denote
the
(unique)
$t$-ary representation of $j$
by
\begin{eqnarray}
\left(s_{h-1}(j),s_{h-2}(j),\ldots,s_0(j)\right)
\in\left\{0,1,\ldots,t-1\right\}^h,
\nonumber
\end{eqnarray}
i.e.,
\begin{eqnarray}
\sum_{\ell=0}^{h-1}\,s_\ell(j)\cdot t^\ell=j.
\label{multiaryrepresentation}
\end{eqnarray}

For all $i$, $j\in\{0,1,\ldots,n-1\}$,
\begin{eqnarray}
\tilde{d}\left(i,i+j\bmod{n}\right)
\stackrel{\text{def.}}{=}
\sum_{k=0}^{h-1}\, d\left(i+\sum_{\ell=h-k}^{h-1}\, s_\ell(j)\cdot
t^\ell\bmod{n},
i+\sum_{\ell=h-1-k}^{h-1}\, s_\ell(j)\cdot t^\ell\bmod{n}\right).
\label{pseudodistance}
\end{eqnarray}
By convention,
empty sums vanish, e.g.,
$\sum_{\ell=h}^{h-1}\,s_\ell(j)\cdot t^\ell=0$.

\begin{lemma}\label{pseudodistanceislarger}
For all $i$, $j\in\{0,1,\ldots,n-1\}$,
$$
d\left(i,i+j\bmod{n}\right)
\le \tilde{d}\left(i,i+j\bmod{n}\right).
$$
\end{lemma}
\begin{proof}
By equation~(\ref{pseudodistance}) and the triangle inequality for $d$,
$$
\tilde{d}\left(i,i+j\bmod{n}\right)
\ge
d\left(i,i+\sum_{\ell=0}^{h-1}\, s_\ell(j)\cdot
t^\ell\bmod{n}\right).
$$
This and
equation~(\ref{multiaryrepresentation})
complete the proof.
\end{proof}

\begin{lemma}
\label{approximationratiolemma}
For all
$\alpha\in\{0,1,\ldots,n-1\}$
with
\begin{eqnarray}
\sum_{j=0}^{n-1}\,\tilde{d}\left(\alpha,\alpha+j\bmod{n}\right)
=
\min_{i=0}^{n-1}\,
\sum_{j=0}^{n-1}\,\tilde{d}\left(i,i+j\bmod{n}\right),
\label{theoptimalpointwithrespecttopseudodistance}
\end{eqnarray}
we have
\begin{eqnarray}
\sum_{j=0}^{n-1}\, d\left(\alpha,j\right)
\le 2h
\cdot
\min_{i=0}^{n-1}\,
\sum_{j=0}^{n-1}\, d\left(i,j\right).
\label{approximationinequality}
\end{eqnarray}
\end{lemma}
\begin{proof}
Let $\bs{u}$
be
a
uniformly random
element
of $\{0,1,\ldots,n-1\}$ and
\begin{eqnarray}
i'=\mathop{\mathrm{argmin}}_{i=0}^{n-1}\, \sum_{j=0}^{n-1}\,
d\left(i,j\right),
\label{startofequations}
\end{eqnarray}
breaking ties arbitrarily.
It is easy to see that
\begin{eqnarray}
\sum_{j=0}^{n-1}\, d\left(\alpha,j\right)
=\sum_{j=0}^{n-1}\, d\left(\alpha,\alpha+j\bmod{n}\right).
\nonumber
\end{eqnarray}
Furthermore,
\begin{eqnarray}
\sum_{j=0}^{n-1}\, d\left(\alpha,\alpha+j\bmod{n}\right)
\nonumber
&\stackrel{\text{Lemma~\ref{pseudodistanceislarger}}}{\le}&
\sum_{j=0}^{n-1}\,
\tilde{d}\left(\alpha,\alpha+j\bmod{n}\right)
\nonumber\\
&\stackrel{\text{equation~(\ref{theoptimalpointwithrespecttopseudodistance})}}{\le}&
\mathop{\mathrm E}\left[\,
\sum_{j=0}^{n-1}\,
\tilde{d}\left(\bs{u},\bs{u}+j\bmod{n}\right)
\,\right].
\nonumber
\end{eqnarray}
By equation~(\ref{pseudodistance}),
\begin{eqnarray}
&&
\mathop{\mathrm E}\left[\,
\sum_{j=0}^{n-1}\,
\tilde{d}\left(\bs{u},\bs{u}+j\bmod{n}
\right)\,\right]
\nonumber\\
&=&
\mathop{\mathrm E}\left[\,
\sum_{j=0}^{n-1}\,
\sum_{k=0}^{h-1}\, d\left(\bs{u}+\sum_{\ell=h-k}^{h-1}\, s_\ell(j)
\cdot t^\ell\bmod{n},
\bs{u}+\sum_{\ell=h-1-k}^{h-1}\, s_\ell(j)\cdot t^\ell\bmod{n}\right)
\,\right].
\nonumber
\end{eqnarray}

Finally,
{\small 
\begin{eqnarray}
&&\mathop{\mathrm E}\left[\,
\sum_{j=0}^{n-1}\,
\sum_{k=0}^{h-1}\, d\left(\bs{u}+\sum_{\ell=h-k}^{h-1}\, s_\ell(j)
\cdot t^\ell\bmod{n},
\bs{u}+\sum_{\ell=h-1-k}^{h-1}\, s_\ell(j)\cdot t^\ell\bmod{n}\right)
\,\right]
\nonumber\\
&\le&
\mathop{\mathrm E}\left[\,
\sum_{j=0}^{n-1}\,
\sum_{k=0}^{h-1}\, d\left(i',
\bs{u}+\sum_{\ell=h-k}^{h-1}\, s_\ell(j)
\cdot t^\ell\bmod{n}\right)
+d\left(i',
\bs{u}+\sum_{\ell=h-1-k}^{h-1}\, s_\ell(j)\cdot t^\ell\bmod{n}\right)
\,\right]
\nonumber\\
&=&
\sum_{j=0}^{n-1}\,
\sum_{k=0}^{h-1}\,
\mathop{\mathrm E}\left[\,
d\left(i',
\bs{u}+\sum_{\ell=h-k}^{h-1}\, s_\ell(j)
\cdot t^\ell\bmod{n}\right)
\,\right]
+\mathop{\mathrm E}\left[\,d\left(i',
\bs{u}+\sum_{\ell=h-1-k}^{h-1}\, s_\ell(j)\cdot t^\ell\bmod{n}\right)
\,\right]
\nonumber\\
&=&
\sum_{j=0}^{n-1}\,
\sum_{k=0}^{h-1}\,
\left(
\mathop{\mathrm E}\left[\,
d\left(i',
\bs{u}\right)
\,\right]
+\mathop{\mathrm E}\left[\,d\left(i',
\bs{u}\right)
\,\right]
\right)\nonumber\\
&=& 2nh \cdot \mathop{\mathrm E}\left[\,d\left(i',
\bs{u}\right)\,\right],
\label{endofequations}
\end{eqnarray}
}
where the
inequality follows from the triangle inequality for $d$, and the
second-to-last equality is true because
$\bs{u}+\sum_{\ell=h-k}^{h-1}\, s_\ell(j)\cdot t^\ell\bmod{n}$
distributes uniformly at random over $\{0,1,\ldots,n-1\}$
for {\em any} $j\in \{0,1,\ldots,n-1\}$ and $k\in\{0,1,\ldots,h\}$.
Inequalities~(\ref{startofequations})--(\ref{endofequations})
imply
inequality~(\ref{approximationinequality}).
\end{proof}

For a predicate $P$,
let $\chi[\,P\,]=1$ if $P$ is true and $\chi[\,P\,]=0$ otherwise.
Define
$(s'_{h-1},s'_{h-2},\ldots,s'_0)\in\{0,1,\ldots,t-1\}^h$
to be
the $t$-ary representation of $n-1$.
So $\sum_{r=0}^{h-1}\,s'_r\cdot t^r=n-1$.
For $i\in\{0,1,\ldots,n-1\}$ and $m\in\{0,1,\ldots,h-1\}$,
{\small 
\begin{eqnarray}
f\left(i,m\right)
&\stackrel{\text{def.}}{=}&
\sum_{s_m,s_{m-1},\ldots,s_0=0}^{t-1}\,
\chi\left[\sum_{r=0}^m\, s_r\cdot t^r\le \sum_{r=0}^m\, s'_r\cdot t^r
\right]\nonumber\\
&\cdot&\sum_{k=0}^m\,
d\left(i+\sum_{\ell=m+1-k}^m\, s_\ell \cdot t^\ell \bmod{n},
i+\sum_{\ell=m-k}^m\, s_\ell \cdot t^\ell \bmod{n}
\right),
\label{subsumlessthanorequalto}\\
g\left(i,m\right)
&\stackrel{\text{def.}}{=}&
\sum_{s_m,s_{m-1},\ldots,s_0=0}^{t-1}\,
\sum_{k=0}^m\,
d\left(i+\sum_{\ell=m+1-k}^m\, s_\ell \cdot t^\ell \bmod{n},
i+\sum_{\ell=m-k}^m\, s_\ell \cdot t^\ell \bmod{n}
\right).
\,\,\,\,\,\,\,\,\,\,\label{subsumlessthan}
\end{eqnarray}
}
Clearly,
\begin{eqnarray}
f\left(i,0\right)
&=&
\sum_{s_0=0}^{s'_0}\,
d\left(i,i+s_0\bmod{n}\right),
\label{thebasecaseofthefirstfunction}\\
g\left(i,0\right)
&=&
\sum_{s_0=0}^{t-1}\,
d\left(i,i+s_0\bmod{n}\right).
\label{thebasecaseofthesecondfunction}
\end{eqnarray}

\begin{lemma}\label{theDPresultisthesumofpseudodistanceslemma}
For all $i\in\{0,1,\ldots,n-1\}$,
\begin{eqnarray}
f\left(i,h-1\right)
=\sum_{j=0}^{n-1}\,
\tilde{d}\left(i,i+j\bmod{n}
\right).
\nonumber
\end{eqnarray}
\end{lemma}
\begin{proof}
As $\sum_{r=0}^{h-1}\,s'_r\cdot t^r= n-1$,
\begin{eqnarray}
f\left(i,h-1\right)
&=&
\sum_{s_{h-1},s_{h-2},\ldots,s_0=0}^{t-1}\,
\chi\left[\,\sum_{r=0}^{h-1}\,s_r\cdot t^r\le n-1\,\right]
\nonumber\\
&\cdot&
\sum_{k=0}^{h-1}\,
d\left(
i+\sum_{\ell=h-k}^{h-1}\,s_\ell\cdot t^\ell \bmod{n},
i+\sum_{\ell=h-1-k}^{h-1}\,s_\ell\cdot t^\ell \bmod{n}
\right).
\,\,\,\,\,\label{thethingwewantintheformofwhatwecancompute}
\end{eqnarray}
By the existence and uniqueness of a $t$-ary representation
of each
$j\in\{0,1,\ldots,n-1\}$,
\begin{eqnarray}
&&\sum_{j=0}^{n-1}\,
\sum_{k=0}^{h-1}\,
d\left(
i+\sum_{\ell=h-k}^{h-1}\,s_\ell(j)\cdot t^\ell \bmod{n},
i+\sum_{\ell=h-1-k}^{h-1}\,s_\ell(j)\cdot t^\ell \bmod{n}
\right)
\nonumber\\
&=&
\sum_{s_{h-1},s_{h-2},\ldots,s_0=0}^{t-1}\,
\chi\left[\,\sum_{r=0}^{h-1}\,s_r\cdot t^r\le n-1\,\right]
\nonumber\\
&\cdot&
\sum_{k=0}^{h-1}\,
d\left(
i+\sum_{\ell=h-k}^{h-1}\,s_\ell\cdot t^\ell \bmod{n},
i+\sum_{\ell=h-1-k}^{h-1}\,s_\ell\cdot t^\ell \bmod{n}
\right).
\label{justanequation}
\end{eqnarray}
%
Equations~(\ref{pseudodistance})~and~(\ref{thethingwewantintheformofwhatwecancompute})--(\ref{justanequation})
complete the proof.
\end{proof}


\begin{lemma}
\label{recurrencelemma1}
For all $i\in\{0,1,\ldots,n-1\}$ and $m\in\{1,2,\ldots,h-1\}$,
\begin{eqnarray}
g\left(i,m
\right)
&=&
t^m\sum_{s_m=0}^{t-1}\,
d\left(i,i+s_m\cdot t^m\bmod{n}\right)\nonumber\\
&+&\sum_{s_m=0}^{t-1}\,
g\left(i+s_m\cdot t^m\bmod{n},m-1
\right).
\nonumber
\end{eqnarray}
\end{lemma}
\begin{proof}
By equation~(\ref{subsumlessthan}),
{\footnotesize 
\begin{eqnarray}
&&g\left(i,m\right)\nonumber\\
&=&
\sum_{s_m=0}^{t-1}\,
\sum_{s_{m-1},s_{m-2},\ldots,s_0=0}^{t-1}\,
\left(d\left(i,i+s_m\cdot t^m\bmod{n}
\right)
\vphantom{\sum_{\ell=m-1-k}^{m-2}\,s_\ell\cdot t^\ell}
\right.\nonumber\\
&&\left.
+\sum_{k=0}^{m-1}\,
d\left(i+s_m\cdot t^m+\sum_{\ell=m-k}^{m-1}\,s_\ell\cdot t^\ell
\bmod{n},
i+s_m\cdot t^m+\sum_{\ell=m-1-k}^{m-1}\,s_\ell\cdot t^\ell\bmod{n}
\right)
\right),\nonumber\\
&&g\left(i+s_m\cdot t^m\bmod{n},m-1\right)\nonumber\\
&=&
\sum_{s_{m-1},s_{m-2},\ldots,s_0=0}^{t-1}\,
\sum_{k=0}^{m-1}\,
d\left(i+s_m\cdot t^m+\sum_{\ell=m-k}^{m-1}\,s_\ell\cdot t^\ell
\bmod{n},
i+s_m\cdot t^m+\sum_{\ell=m-1-k}^{m-1}\,s_\ell\cdot t^\ell\bmod{n}
\right)\nonumber
\end{eqnarray}
}
for $s_m\in\{0,1,\ldots,t-1\}$.
Furthermore,
\begin{eqnarray}
\sum_{s_m=0}^{t-1}\,
\sum_{s_{m-1},s_{m-2},\ldots,s_0=0}^{t-1}\,
d\left(i,i+s_m\cdot t^m\bmod{n}\right)
=
t^m\sum_{s_m=0}^{t-1}\,
d\left(i,i+s_m\cdot t^m\bmod{n}
\right).
\nonumber
\end{eqnarray}
\end{proof}

\begin{lemma}
\label{recurrencelemma2}
For all $i\in\{0,1,\ldots,n-1\}$ and $m\in\{1,2,\ldots,h-1\}$,
\begin{eqnarray}
f\left(i,m\right)
&=&
\left(1+\sum_{r=0}^{m-1}\, s'_r\cdot t^r\right) d\left(i,i+s'_m\cdot
t^m\bmod{n}\right)\nonumber\\
&+&
t^m
\sum_{s_m=0}^{s'_m-1}\,
d\left(i,i+s_m\cdot t^m\bmod{n}\right)\nonumber\\
&+&
f\left(i+s'_m\cdot t^m\bmod{n},m-1\right)\nonumber\\
&+&\sum_{s_m=0}^{s'_m-1}\, g\left(i+s_m\cdot t^m\bmod{n},m-1
\right).
\nonumber
\end{eqnarray}
\end{lemma}
\begin{proof}
Observe
the following
for all
$s_m$, $s_{m-1}$, $\ldots$, $s_0\in\{0,1,\ldots,t-1\}$:
\begin{enumerate}[(i)]
\item\label{item1}
If $s_m=s'_m$, then
$\sum_{r=0}^m\, s_r\cdot t^r\le \sum_{r=0}^m\, s'_r\cdot t^r$
if and only if
$\sum_{r=0}^{m-1}\, s_r\cdot t^r\le \sum_{r=0}^{m-1}\, s'_r\cdot t^r$;
\item\label{item2}
If $s_m<s'_m$, then
$\sum_{r=0}^m\, s_r\cdot t^r< \sum_{r=0}^m\, s'_r\cdot t^r$;
\item\label{item3}
If $s_m>s'_m$, then
$\sum_{r=0}^m\, s_r\cdot t^r> \sum_{r=0}^m\, s'_r\cdot t^r$.
\end{enumerate}

\comment{ 
For convenience,
$$
Q\stackrel{\text{def.}}{=}
d\left(i+\sum_{\ell=0}^{k-1}\, s_\ell \cdot t^\ell \bmod{n},
i+\sum_{\ell=0}^k\, s_\ell \cdot t^\ell \bmod{n}
\right).
$$
}

We have
{\footnotesize 
\begin{eqnarray}
&&
f\left(i,m\right)
\label{thethingwewanttorecurseon}\\
&\stackrel{\text{equation~(\ref{subsumlessthanorequalto})}}{=}&
\sum_{s_m=0}^{t-1}\,
\sum_{s_{m-1},s_{m-2},\ldots,s_0=0}^{t-1}\,
\chi\left[\sum_{r=0}^m\, s_r\cdot t^r\le \sum_{r=0}^m\, s'_r\cdot t^r
\right]
\cdot
\left(d\left(i,i+s_m\cdot t^m\bmod{n}\right)
\vphantom{\sum_{\ell=m-2-k}^{m-2}\, s_\ell\cdot t^\ell\bmod{n}}
\right.\nonumber\\
&&\left.+\sum_{k=0}^{m-1}\, d\left(i+s_m\cdot t^m+\sum_{\ell=m-k}^{m-1}\,
s_\ell\cdot t^\ell\bmod{n}, i+s_m\cdot t^m+\sum_{\ell=m-1-k}^{m-1}\,
s_\ell\cdot t^\ell\bmod{n}
\right)
\right)\nonumber\\
&\stackrel{\text{item~(\ref{item3})}}{=}&
\sum_{s_m=0}^{s'_m}\,
\sum_{s_{m-1},s_{m-2},\ldots,s_0=0}^{t-1}\,
\left(
\chi\left[\left(s_m=s'_m\right)\land
\left(\sum_{r=0}^m\, s_r\cdot t^r\le \sum_{r=0}^m\, s'_r\cdot t^r
\right)
\right]\right.\nonumber\\
&&\left.+\chi\left[\left(s_m<s'_m\right)\land
\left(\sum_{r=0}^m\, s_r\cdot t^r\le \sum_{r=0}^m\, s'_r\cdot t^r
\right)
\right]
\right)
\cdot
\left(d\left(i,i+s_m\cdot t^m\bmod{n}\right)
\vphantom{\sum_{\ell=m-2-k}^{m-2}\, s_\ell\cdot t^\ell\bmod{n}}
\right.\nonumber\\
&&\left.+\sum_{k=0}^{m-1}\, d\left(i+s_m\cdot t^m+\sum_{\ell=m-k}^{m-1}\,
s_\ell\cdot t^\ell\bmod{n}, i+s_m\cdot t^m+\sum_{\ell=m-1-k}^{m-1}\,
s_\ell\cdot t^\ell\bmod{n}
\right)
\right)\nonumber\\
&\stackrel{\text{item~(\ref{item1})}}{=}&
\sum_{s_m=0}^{s'_m}\,
\sum_{s_{m-1},s_{m-2},\ldots,s_0=0}^{t-1}\,
\left(
\chi\left[\left(s_m=s'_m\right)\land
\left(\sum_{r=0}^{m-1}\, s_r\cdot t^r\le \sum_{r=0}^{m-1}\, s'_r\cdot t^r
\right)
\right]\right.\nonumber\\
&&\left.+\chi\left[\left(s_m<s'_m\right)\land
\left(\sum_{r=0}^m\, s_r\cdot t^r\le \sum_{r=0}^m\, s'_r\cdot t^r
\right)
\right]
\right)
\cdot
\left(d\left(i,i+s_m\cdot t^m\bmod{n}\right)
\vphantom{\sum_{\ell=m-2-k}^{m-2}\, s_\ell\cdot t^\ell\bmod{n}}
\right.\nonumber\\
&&\left.+\sum_{k=0}^{m-1}\, d\left(i+s_m\cdot t^m+\sum_{\ell=m-k}^{m-1}\,
s_\ell\cdot t^\ell\bmod{n}, i+s_m\cdot t^m+\sum_{\ell=m-1-k}^{m-1}\,
s_\ell\cdot t^\ell\bmod{n}
\right)
\right)\nonumber\\
&\stackrel{\text{item~(\ref{item2})}}{=}&
\sum_{s_m=0}^{s'_m}\,
\sum_{s_{m-1},s_{m-2},\ldots,s_0=0}^{t-1}\,
\left(
\chi\left[\left(s_m=s'_m\right)\land
\left(\sum_{r=0}^{m-1}\, s_r\cdot t^r\le \sum_{r=0}^{m-1}\, s'_r\cdot t^r
\right)
\right]\right.\nonumber\\
&&\left.+\chi\left[s_m<s'_m
\right]
\vphantom{\sum_{r=0}^{m-2}\, s'_r\cdot t^r}
\right)
\cdot
\left(d\left(i,i+s_m\cdot t^m\bmod{n}\right)
\vphantom{\sum_{\ell=m-2-k}^{m-2}\, s_\ell\cdot t^\ell\bmod{n}}
\right.\nonumber\\
&&\left.+\sum_{k=0}^{m-1}\, d\left(i+s_m\cdot t^m+\sum_{\ell=m-k}^{m-1}\,
s_\ell\cdot t^\ell\bmod{n}, i+s_m\cdot t^m+\sum_{\ell=m-1-k}^{m-1}\,
s_\ell\cdot t^\ell\bmod{n}
\right)
\right).
\nonumber
\end{eqnarray}
}

By equation~(\ref{subsumlessthanorequalto}),
{\small 
\begin{eqnarray}
&&f\left(i+s'_m\cdot t^m\bmod{n},m-1
\right)\nonumber\\
&=& \sum_{s_{m-1},s_{m-2},\ldots,s_0=0}^{t-1}\,
\chi\left[\,\sum_{r=0}^{m-1}\,s_r\cdot t^r
\le \sum_{r=0}^{m-1}\,s'_r\cdot t^r
\,\right]\nonumber\\
&\cdot&
\sum_{k=0}^{m-1}\,
d\left(i+s'_m\cdot t^m+\sum_{\ell=m-k}^{m-1}\,s_\ell\cdot
t^\ell\bmod{n},
i+s'_m\cdot t^m+\sum_{\ell=m-1-k}^{m-1}\,s_\ell\cdot
t^\ell\bmod{n}
\right)\nonumber\\
&=&
\sum_{s_m=0}^{s'_m}\,
\sum_{s_{m-1},s_{m-2},\ldots,s_0=0}^{t-1}\,
\chi\left[\,\left(s_m=s'_m\right)\land
\left(\sum_{r=0}^{m-1}\,s_r\cdot t^r
\le \sum_{r=0}^{m-1}\,s'_r\cdot t^r\right)
\,\right]\nonumber\\
&\cdot&\sum_{k=0}^{m-1}\,
d\left(i+s_m\cdot t^m+\sum_{\ell=m-k}^{m-1}\,s_\ell\cdot
t^\ell\bmod{n},
i+s_m\cdot t^m+\sum_{\ell=m-1-k}^{m-1}\,s_\ell\cdot
t^\ell\bmod{n}
\right).
\nonumber
\end{eqnarray}
}
By equation~(\ref{subsumlessthan}),
{\small 
\begin{eqnarray}
&&\sum_{s_m=0}^{s'_m-1}\,
g\left(i+s_m\cdot t^m\bmod{n},m-1
\right)\nonumber\\
&=&
\sum_{s_m=0}^{s'_m}\,
\sum_{s_{m-1},s_{m-2},\ldots,s_0=0}^{t-1}\,
\chi\left[\,s_m<s'_m\,\right]
\nonumber\\
&\cdot&
\sum_{k=0}^{m-1}\,
d\left(i+s_m\cdot t^m+\sum_{\ell=m-k}^{m-1}\, s_\ell\cdot t^\ell
\bmod{n}, i+s_m\cdot t^m+\sum_{\ell=m-1-k}^{m-1}\, s_\ell\cdot
t^\ell \bmod{n}
\right).
\nonumber
\end{eqnarray}
}

Because each number in $\{0,1,\ldots,\sum_{r=0}^{m-1}\,s'_r\cdot t^r\}$
can be written
uniquely
as
$\sum_{r=0}^{m-1}\,s_r\cdot t^r$,
where $s_{m-1}$, $s_{m-2}$, $\ldots$, $s_0\in\{0,1,\ldots,t-1\}$,
\begin{eqnarray}
&&\sum_{s_m=0}^{s'_m}\,
\sum_{s_{m-1},s_{m-2},\ldots,s_0=0}^{t-1}\,
\chi\left[\,\left(s_m=s'_m\right)\land
\left(\sum_{r=0}^{m-1}\,s_r\cdot t^r \le \sum_{r=0}^{m-1}\, s'_r\cdot t^r
\right)
\,\right]\nonumber\\
&\cdot&
d\left(i,i+s_m\cdot t^m\bmod{n}\right)\nonumber\\
&=&
\sum_{s_m=0}^{s'_m}\,
\left(1+\sum_{r=0}^{m-1}\, s'_r\cdot t^r\right)
\cdot
\chi\left[\,s_m=s'_m\,\right]
\cdot d\left(i,i+s_m\cdot t^m\bmod{n}\right)\nonumber\\
&=&
\left(1+\sum_{r=0}^{m-1}\, s'_r\cdot t^r\right)
d\left(i,i+s'_m\cdot t^m\bmod{n}\right).
\nonumber
\end{eqnarray}
Finally,
\begin{eqnarray}
&&\sum_{s_m=0}^{s'_m}\,
\sum_{s_{m-1},s_{m-2},\ldots,s_0=0}^{t-1}\,
\chi\left[\,s_m<s'_m
\,\right]
\cdot d\left(i,i+s_m\cdot t^m\bmod{n}
\right)\nonumber\\
&=& t^m \sum_{s_m=0}^{s'_m-1}\, d\left(i,i+s_m\cdot
t^m\bmod{n}
\right).
\label{lastoftediousequations}
\end{eqnarray}
Equations~(\ref{thethingwewanttorecurseon})--(\ref{lastoftediousequations})
complete the proof.
\end{proof}

\begin{figure}
\begin{algorithmic}[1]
\STATE $t\leftarrow\lceil n^{1/h}\rceil$;
\STATE Find
the $t$-ary representation
of $n-1$, denoted
$(s'_{h-1},s'_{h-2},\ldots,s'_0)\in\{0,1,\ldots,t-1\}^h$;
\FOR{$i\in\{0,1,\ldots,n-1\}$}
  \STATE $f[i][0]\leftarrow \sum_{s_0=0}^{s'_0}\, d(i,i+s_0\bmod{n})$;
  \STATE $g[i][0]\leftarrow \sum_{s_0=0}^{t-1}\, d(i,i+s_0\bmod{n})$;
\ENDFOR
\FOR{$m=1$ up to $h-1$}
  \FOR{$i\in\{0,1,\ldots,n-1\}$}
    \STATE $f[i][m]\leftarrow (1+\sum_{r=0}^{m-1}\,s'_r\cdot t^r)\,
d(i,i+s'_m\cdot t^m\bmod{n})$;
    \STATE $f[i][m]\leftarrow f[i][m]+t^m\sum_{s_m=0}^{s'_m-1}\,
d(i,i+s_m\cdot t^m\bmod{n})$;
    \STATE $f[i][m]\leftarrow f[i][m]+f[i+s'_m\cdot t^m\bmod{n}][m-1]$;
    \STATE $f[i][m]\leftarrow f[i][m]+\sum_{s_m=0}^{s'_m-1}\,
g[i+s_m\cdot t^m\bmod{n}][m-1]$;
    \STATE $g[i][m]\leftarrow t^m\sum_{s_m=0}^{t-1}\,
d(i,i+s_m\cdot t^m\bmod{n})$;
    \STATE $g[i][m]\leftarrow g[i][m]+\sum_{s_m=0}^{t-1}\,
g[i+s_m\cdot t^m\bmod{n}][m-1]$;
  \ENDFOR
\ENDFOR
\STATE Output $\mathop{\mathrm{argmin}}_{i=0}^{n-1}\, f[i][h-1]$,
breaking ties arbitrarily;
\end{algorithmic}
\caption{Algorithm {\sf find-median} with input
a
metric space $(\{0,1,\ldots,n-1\},d)$ and $h\in\mathbb{Z}^+\setminus\{1\}$.}
\label{mainalgorithm}
\end{figure}

\comment{ 
Lemmas~\ref{recurrencelemma1}--\ref{recurrencelemma2}
allow us to compute
$f(\cdot)$
by dynamic programming.
}

\begin{lemma}
\label{algorithmapproximationratiolemma}
Algorithm {\sf find-median} in Fig.~\ref{mainalgorithm}
is $(2h)$-approximate for {\sc metric $1$-median}.
\end{lemma}
\begin{proof}
By
equations~(\ref{thebasecaseofthefirstfunction})--(\ref{thebasecaseofthesecondfunction}),
lines~4--5 of {\sf find-median}
compute $f(i,0)$
and $g(i,0)$.
Then, by Lemmas~\ref{recurrencelemma1}--\ref{recurrencelemma2},
$f(i,m)$
and $g(i,m)$
can be found by dynamic programming as in lines~9--14.
So line~17 outputs
$\mathop{\mathrm{argmin}}_{i=0}^{n-1}\,f(i,h-1)$,
which equals
$$\mathop{\mathrm{argmin}}_{i=0}^{n-1}\,\sum_{j=0}^{n-1}\,
\tilde{d}\left(i,i+j\bmod{n}\right)$$
by Lemma~\ref{theDPresultisthesumofpseudodistanceslemma}.
Now Lemma~\ref{approximationratiolemma}
gives the approximation ratio of $2h$.
\end{proof}

We now state our main theorem.

\begin{theorem}
{\sc Metric $1$-median} has a
deterministic
$O(h n^{1+1/h})$-time
$(2h)$-approximation
nonadaptive
algorithm for each
$h\in\mathbb{Z}^+\setminus\{1\}$.
\end{theorem}
\begin{proof}
Clearly, {\sf find-median} is deterministic and nonadaptive.
Furthermore, it is $(2h)$-approximate for {\sc metric $1$-median}
by Lemma~\ref{algorithmapproximationratiolemma}.
As $s'_i\le t-1$ for all $i\in\{0,1,\ldots,h-1\}$,
the loop in lines~3--6 of {\sf find-median} takes $O(nt)$ time.
By precomputing $t^i$ and $\sum_{r=0}^i\, s'_r\cdot t^r$ for all
$i\in\{0,1,\ldots,h-1\}$, each iteration of the loop in lines~8--15
takes $O(t)$ time.
\end{proof}


\bibliographystyle{plain}
\bibliography{median_refinement}

\noindent

\end{document}